\documentclass[letterpaper,11pt]{article}

\def\authorNote{} 
\usepackage[utf8]{inputenc}
\usepackage{amsmath,amssymb}
\usepackage{framed,caption,color,url}
\usepackage{standalone}
\usepackage{times}
\usepackage{mdframed}
\usepackage{algorithm2e,algorithmicx,algpseudocode,tikz,wrapfig}
\usepackage{graphicx}               
\usepackage{cite} 
\usepackage{nicefrac} 
\usepackage{paralist} 
\usepackage{boxedminipage} 
\usepackage[pdfstartview=FitH,colorlinks,linkcolor=blue,filecolor=blue,citecolor=blue,urlcolor=blue,pagebackref=true]{hyperref}
\usepackage{upgreek}
\usepackage{centernot}
\usepackage{mathtools}
\usepackage{cleveref} 

\ifdefined\LLNCS
\else
\usepackage{fullpage}
\usepackage{amsthm}
\fi

\newcommand{\ugets}{\xleftarrow{\$}}
\newcommand{\U}{\mathfrak{U}}
\newcommand{\minus}{\scalebox{0.7}{-}}

\newcommand{\fhat}[2]{\ifthenelse{\equal{#2}{}}{\hat{f}[#1]}{\ifthenelse{\equal{#2}{0}}{\hat{f}[\emptyset]}{\hat{f}[#1_{\leq #2}]}}}
\newcommand{\gain}[2]{\ifthenelse{\equal{#2}{}}{g[#1]}{g[#1_{\leq #2}]}}

\newcommand{\pr}[2][]{\Pr_{\ifthenelse{\isempty{#1}}{}{{#1}}}\left[{#2}\right]}

\usepackage{dsfont}

\newcommand{\e}{\mathrm{e}}

\newcommand{\remove}[1]{}

\newcommand{\etal}{et~al.\ }

\newcommand{\eg}{e.g.,\ }
\newcommand{\wrt} {with respect to\ }

\newcommand{\floor}[1]{\lfloor #1 \rfloor}

\newcommand{\bra}[1]{\langle#1\rvert}

\newcommand{\ket}[1]{\lvert#1\rangle}
\newcommand{\set}[1]{\{ #1 \}}

\newcommand{\C}{\mathbb{C}}

\newcommand{\N}{{\mathbb N}}

\newcommand{\cA}{{\mathcal A}}

\newcommand{\bfy}{\mathbf{y}}

\usepackage{bm}

\newcommand{\sfA}{\mathsf{A}}

\newcommand{\sfC}{\mathsf{C}}

\newcommand{\sfO}{\mathsf{O}}

\newcommand{\sfW}{\mathsf{W}}

\newcommand{\veps}{\varepsilon}

\newcommand{\poly}{\operatorname{poly}}

\ifdefined\LLNCS
\newtheorem{observation}[theorem]{Observation}
\newtheorem{fact}[theorem]{Fact}

\newtheorem{construction}[theorem]{Construction}
\newtheorem{algorithm1}[theorem]{Algorithm}
\newtheorem{assumption}[theorem]{Assumption}

\makeatletter
\renewcommand{\@Opargbegintheorem}[4]{%
  #4\trivlist\item[\hskip\labelsep{#3#2\@thmcounterend}]}
\makeatother

\else

\newtheorem{theorem}{Theorem}[section]
\theoremstyle{plain}

\newtheorem{lemma}[theorem]{Lemma}
\newtheorem{corollary}[theorem]{Corollary}

\newtheorem{fact}[theorem]{Fact}
\theoremstyle{definition}

\newtheorem{definition}[theorem]{Definition}

\theoremstyle{definition}
\newtheorem{remark}[theorem]{Remark}

\fi

\ifdefined\LLNCS
\else

\newcommand{\sdotfill}{\textcolor[rgb]{0.8,0.8,0.8}{\dotfill}} 

\makeatletter
\def\th@protocol{%
\normalfont 
\setbeamercolor{block title example}{bg=orange,fg=white}
\setbeamercolor{block body example}{bg=orange!20,fg=black}
\def\inserttheoremblockenv{exampleblock}
}
\makeatother
\theoremstyle{protocol}
\newtheorem{proto}[theorem]{Protocol}
\newtheorem{protoc}[theorem]{Protocol}

\fi

\newcommand{\namedref}[2]{#1~\ref{#2}}

\newcommand{\torestate}[3]{%
\expandafter \def \csname BBRESTATE #2 \endcsname{#3}
\theoremstyle{plain}
\newtheorem{BBRESTATETHMNUM#2}[theorem]{#1}
\begin{BBRESTATETHMNUM#2}\label{#2}\csname BBRESTATE #2 \endcsname   \end{BBRESTATETHMNUM#2}
\newtheorem*{BBRESTATETHMNONNUM#2}{\namedref{#1}{#2}}
}

\newcommand{\restate}[1]{\begin{BBRESTATETHMNONNUM#1}[Restated] \csname BBRESTATE #1 \endcsname
\end{BBRESTATETHMNONNUM#1}}

\renewcommand\qedsymbol{$\diamond$}

\usepackage[utf8]{inputenc}
\usepackage[top=1in, bottom=1in, left=1.25in, right=1.25in]{geometry}
\usepackage{pst-node}
\usepackage{tikz-cd}
\usepackage{comment}
\allowdisplaybreaks

\ifdefined\authorNote
\newcommand{\YT}[1]{{\color{brown} [{YT:} #1]}}
\else
\newcommand{\YT}[1]{}
\fi

\newcommand{\email}[1]{\href{mailto:#1}{#1}}

\title{A Note on Quantum Phase Estimation}

\author{
Yao-Ting Lin \thanks{UCSB, \email{yao-ting\_lin@ucsb.edu}. Part of the work was done when working at Academia Sinica.}
}

\date{}

\begin{document}

\maketitle

\begin{abstract}
    In this work, we study the phase estimation problem. 
    We show an alternative, simpler and self-contained proof of query lower bounds.
    Technically, compared to the previous proofs~\cite{nayak1999quantum, bessen2005lower}, our proof is considerably elementary. Specifically, our proof consists of basic linear algebra without using the knowledge of Boolean function analysis and adversary methods. 
    Qualitatively, our bound is tight in the \emph{low success probability} regime and offers a more fine-grained trade-off. In particular, we prove that for any $\veps>0,p \geq 0$, every algorithm requires at least $\Omega(p/\veps)$ queries to obtain an $\veps$-approximation for the phase with probability at least $p$. However, the existing bounds hold only when $p > 1/2$.
    Quantitatively, our bound is tight since it matches the well-known phase estimation algorithm of Cleve, Ekert, Macchiavello, and Mosca~\cite{cleve1998quantum} which requires $O(1/\veps)$ queries to obtain an $\veps$-approximation with a constant probability. Following the derivation of the lower bound in our framework, we give a new and intuitive interpretation of the phase estimation algorithm of~\cite{cleve1998quantum}, which might be of independent interest.
\end{abstract}
\section{Introduction}
\paragraph{Background.}
First proposed by Kitaev~\cite{kitaev1995quantum}, phase estimation is one of the most fundamental and widely-used subroutines in various quantum algorithms, \eg  Shor's algorithm~\cite{shor1999polynomial}, HHL algorithm~\cite{harrow2009quantum}, Hamiltonian simulation~\cite{berry2009black, childs2010relationship}, quantum approximate counting~\cite{brassard2002quantum}, quantum random walk~\cite{magniez2011search}, the quantization of Markov chains~\cite{szegedy2004quantum} and many others.

In the phase estimation problem, the algorithm is given oracle access to an unknown controlled-unitary $c\minus U$ and a copy of the corresponding eigenstate $\ket{u}$ such that $U\ket{u}=\e^{2\pi i\theta}\ket{u}$, where $\theta \in [0,1)$. The goal of the algorithm is to output an estimation of $\theta$.
The well-known, ``textbook version''~\cite{nielsen2010quantum} phase estimation algorithm of Cleve \etal~\cite{cleve1998quantum} that is based on (inverse) quantum Fourier transform can approximate the phase $\theta$ within an additive error $\veps$ with probability $\Omega(1)$ by making $O(1/\veps)$ queries.

It turns out that the algorithm is optimal in terms of \emph{query complexity}. The lower bound is obtained by reducing the quantum counting problem to the amplitude estimation problem and then reducing the amplitude estimation problem to the phase estimation problem.
The query lower bound of quantum counting was first proved by Nayak and Wu~\cite{nayak1999quantum} using the polynomial method \cite{beals2001quantum}. The work of Bessen~\cite{bessen2005lower} used a different yet arguably complicated technique to obtain the same result. From the above-mentioned works, it was known that to achieve an $\veps$-approximation of $\theta$ with any constant probability greater than $1/2$ requires at least $\Omega(1/\veps)$ queries.
\paragraph{Our Contribution.} 
In this work, we have the following main result:
\begin{theorem}[Corollary~\ref{cor:PE}, restated]
\label{thm:main:restated:1}
For any $\veps\in[1,0)$ and $p\in[1,0]$, any oracle-aided quantum algorithm that has access to an unknown controlled-unitary oracle and has the corresponding eigenstate requires at least $\Omega(p/\veps)$ queries to output $\Tilde{\theta}$ such that $\Pr[|\Tilde{\theta}-\theta|\le \veps]\ge p$.
\end{theorem}

Compared to the previous results, our bound is more fine-grained as it also holds when $p\le1/2$. It lower-bounds the query complexity of obtaining an $\veps$-approximation with \emph{any} probability $p\ge0$. In many settings, \eg cryptography, even a non-negligible\footnote{A function $f(n)$ is negligible if $f(n) = 1/n^{\omega(1)}$.} probability is considered as fatal to the security. 
For example, let $p=1/\poly(n)$, then the previous results and techniques are not applicable.
Hence, one may consider the following scenario: suppose the algorithm (adversary) can make at most $q$ queries, then what is the trade-off between the precision $\veps$ and the success probability $p$? We can reinterpret Theorem~\ref{thm:main:restated:1} in the following: 
\begin{theorem}[Theorem~\ref{thm:main:restated:1}, restated]
For any integer $q \geq 0$, any $q$-query oracle-aided quantum algorithm that has access to an unknown controlled-unitary oracle and has the corresponding eigenstate, the probability $p$ that the algorithm successfully outputs $\Tilde{\theta}$ such that $|\Tilde{\theta}-\theta|\le\veps$ satisfies the trade-off relation $p/\veps=O(q)$.
\end{theorem}

Our bound is tight in terms of the number of queries. As a merit, our proof is elementary and simpler. Compared to the existing proofs which rely on Boolean function analysis and adversary methods, ours consists of merely basic linear algebra.

Furthermore, we consider the \emph{average hardness} where the oracle $U$ is sampled from some distribution. In Corollary~\ref{cor:AvPE}, we show that there exists a distribution $\mu$ over unitary operators such that when $U$ is sampled according to $\mu$, for any algorithm to obtain an $\veps$-approximation with any probability $p$ \emph{on average} still requires $\Omega(p/\veps)$ queries.

Note that for hardness results (lower bounds), the statement is stronger when the algorithm is only 
required to achieve a low success probability averaging over random input. To the best of our knowledge, none of the aforementioned previous results can be trivially extended to the setting in which $p\le1/2$ and the oracle is random\footnote{The the polynomial method~\cite{beals2001quantum} relies on the \emph{approximate degree} $\widetilde{\deg}_\epsilon(f)$ of Boolean functions to lower-bound the number of queries, where $\epsilon$ is the approximation factor. With a low success probability $p$, the factor $\epsilon$ will be close to $1$. However, the approximate degree is then $0$ since the constant function $f(x)=1/2$ is sufficient for the approximation. 
In the adversary method (and its variants)~\cite{ambainis2002quantum, hoyer2007negative}, generally, a decision problem is reduced to the original problem. Thus, in order to make the reduction work, the success probability $p$ needs to be greater than $1/2$. 
In the proof of~\cite{bessen2005lower}, to obtain eqn.~(22) and the inequality above it requires $p>1/2$.}.

Moreover, under the framework of our proof, in Section~\ref{sec:interpretation} we provide a new interpretation of the phase estimation problem. Conceptually, it makes the algorithm in~\cite{cleve1998quantum} more intuitive and understandable.

\paragraph{Technical Overview.}
Our approach is inspired by Zhandry's compressed oracle technique~\cite{zhandry2019record}. First, we introduce the  \emph{phase distinguishing problem} parameterized by $N\in\N$. In the problem, a set of unitary operators which share a common eigenstate $\ket{u}$ are labeled by $y\in\set{0,1,\dots, N-1}$. In the beginning, a unitary $U_y$ with eigenvalue $e^{2\pi yi/N}$ is sampled uniformly from the set and the algorithm is given the oracle access to controlled-$U_y$, denoted by $c\minus U_y$, and its eigenstate $\ket{u}$. The goal of the algorithm is to determine which unitary operator was chosen.

Suppose an algorithm can obtain an $\veps$-approximation of the phase with probability at least $p$. Then there exists a reduction that solves the phase distinguishing problem parameterized by $N = \floor{1/(2\veps)}$ with probability at least $p$ simply by rounding the approximated phase to the closest $y/N$. Consequently, we can derive the query lower bound for the phase estimation problem by showing the query lower bound for the phase distinguishing problem  (Corollary~\ref{cor:PE}). Under the same oracle distribution as in the phase distinguishing problem, we immediately conclude that it is also hard for any algorithm to approximate the phase $\emph{on average}$ over the oracles (Corollary~\ref{cor:AvPE}). 

The main technical part is to prove the hardness of the phase distinguishing problem (Theorem~\ref{thm:main}). In reminiscence of~\cite{zhandry2019record}, our proof precedes by \emph{purifying} the computation and analyzing the algorithm in the \emph{Fourier} basis. 
In particular, besides the registers of the algorithm, we augment an additional register $\mathsf{C}$ that records the choice of $y$. 
Notice that the distribution of $y$ is equivalent to measuring the uniform superposition state in the register $\mathsf{C}$. Since this measurement commutes with the operations performed by the algorithm, it can be deferred to the end. Therefore, the register $\mathsf{C}$ is initialized with uniform superposition which is the zero in the Fourier basis.
Then we introduce the coherent version of the unitary, denoted by $c\minus \U$, which implements $c\minus U_y$ in superposition if the content on register $\mathsf{C}$ is $y$.
As an important observation, when expressing $c\minus\U$ in the Fourier basis, each $c\minus\U$ increases the content of $\mathsf{C}$ at most by $1$.
Given the above, by representing the register $\mathsf{C}$ in the Fourier basis, the register $\mathsf{C}$ can be viewed as a \emph{counter} that records the number of queries. By leveraging this property and applying the Cauchy-Schwarz inequality, we can derive the maximum success probability of the algorithm. 

As a remark, our technique bears some similarities to~\cite{zhandry2019record} in the following sense. In the computational basis, the register $\mathsf{C}$ serves as a \emph{control} register that \emph{writes} the information to the algorithm; while in the Fourier basis, the register $\mathsf{C}$ serves as a \emph{counter} that \emph{records} the behavior of the algorithm.

In the purified view, one can interpret the algorithm's goal as generating a correlation between itself and the register $\mathsf{C}$. Therefore, one can view solving the phase distinguishing problem as manipulating a high-dimensional EPR state. Combining it with a property of (high-dimensional) EPR states, one can easily find the best strategy for solving the phase distinguishing problem. Surprisingly, it turns out that the strategy is the same as the phase estimation algorithm in~\cite{cleve1998quantum}.

\paragraph{Open Problems.}
A number of lower-bound techniques for classical problems have been developed~\cite{bennett1997strengths, beals2001quantum, ambainis2002quantum, aaronson2004quantum, hoyer2007negative, zhandry2019record}. However, techniques for pure \emph{quantum} tasks are relatively rare. We are looking forward to finding more applications of our technique.

\paragraph{Organization.}  In Section~\ref{sec:prelim}, we introduce the definitions and notations. In Section~\ref{sec:proof}, we present our main technical contribution. In Section~\ref{sec:interpretation}, we give an interpretation of our proof.
\section{Preliminaries}\label{sec:prelim}
\subsection{Notations}
For $N\in\N$, by $[N]$ we denote the set $\set{0,1,\dots,N-1}$.
We use bold letters (\eg $\bfy$) to denote random variables.
We use calligraphic letters (\eg $\cA$) to denote algorithms.
For a distribution $\mu$, by $u \gets \mu$ we mean sample $u$ from the distribution $\mu$.
By $y\ugets[N]$ we mean that $y$ is sampled uniformly at random from the set $[N]$. 
We use sans-serif letters (\eg $\mathsf{A},\mathsf{PD_N}$) to denote registers and problems.
By $\|\cdot\|$ we denote the Euclidean norm.
Throughout this work, we use the standard bra-ket notation for quantum objects. For the basics of quantum computing, we refer the readers to \cite{nielsen2010quantum}.

\subsection{Definitions}
\begin{definition}[Quantum Fourier Transform]
    Let $\set{\ket{y}}_{y\in[N]}$ be an orthonormal basis of $\C^N$, which we refer to as the \emph{computational basis}. The \emph{quantum Fourier transform} is defined by the following unitary mapping
    \[
    \mathsf{QFT} \colon \ket{y} \mapsto \ket{\hat{y}} \coloneqq \frac{1}{\sqrt{N}}\sum_{y'\in [N]}\omega_N^{yy'}\ket{y'},
    \]
    where $\omega_N \coloneqq e^{2\pi i/N}$ is the $N$-th root of unity and $\set{\ket{\hat{y}}}_{y\in[N]}$ is called the \emph{Fourier basis}.
    
    We adopt the following notation:
    \[ \begin{tikzcd}
    \ket{y} \arrow{r}{\mathsf{QFT}} \arrow[swap]{d}{\mathsf{QFT}^\dagger} & \ket{\hat{y}} \arrow{d}{\mathsf{QFT}^\dagger} \\%
    \ket{\minus\hat{y}} \arrow{r}{\mathsf{QFT}}& \ket{y}
    \end{tikzcd}
    \]
    where $\ket{\minus\hat{y}}$ denote the complex conjugate of $\ket{\hat{y}}$.
\qedsymbol
\end{definition}

\begin{definition}[Oracle-Aided Quantum Algorithm]
Let $q \geq 0$ be an integer. A $q$-query quantum algorithm $\cA^{(\cdot)}(\cdot)$ given oracle access to $O$ is specified by a sequence of unitary operators $A_0,A_1,\dots, A_q$. 
The final state of the algorithm is defined to be
\[
A_q O \dots A_1 O A_0 \ket{\phi_0}_{\mathsf{A}},
\]
where the subscript $\mathsf{A}$ denote the register of the algorithm and $\ket{\phi_0}$ is the initial state of $\cA$ which possibly depends on the input.

Note that in this work, we focus only on \emph{query} complexity. Thus, the size of $\cA$'s internal workspace can possibly be unbounded and each unitary operator $A_i$ is not necessarily efficiently implementable.
\qedsymbol
\end{definition}

For any unitary operator $U$, by \emph{controlled-$U$} or the \emph{controlled-version of $U$} we refer to the unitary operator $c\minus U \coloneqq \ket{0}\bra{0}\otimes I + \ket{1}\bra{1}\otimes U$. Below we define the (standard) phase estimation problem.
\begin{definition}[Phase Estimation Problem]\label{def:PE}
Given oracle access to a controlled-unitary $c\scalebox{0.9}{-}U$ and its eigenstate $\ket{u}$ such that $U\ket{u} = e^{2\pi i\theta}\ket{u}$, where $\theta \in [0,1)$. Find the phase $\theta$.

For every $\veps>0$ and $p\geq0$, we say an algorithm \emph{$\cA$ $(\veps,p)$-estimates} the phase if for every $(U,\ket{u})$, it holds that
\[
\Pr\left[|\Tilde{\theta}-\theta|\leq \veps \colon \Tilde{\theta}\gets\cA^{c\minus U}(\ket{u}) \right] \geq p,
\] 
where the probability is over the randomness of $\cA$.
\qedsymbol
\end{definition}

\begin{definition}[Average-Input Phase Estimation Problem]
Given oracle access to a unitary $c\minus U$ and its the eigenstate $\ket{u}$ such that $U\ket{u} = e^{2\pi i\theta}\ket{u}$, where $\theta \in [0,1)$ and $U$ is sampled according to some distribution $\mu$ over unitaries that share a common eigenstate $\ket{u}$. Find the phase $\theta$.

For every $\veps>0$ and $p\geq0$ and distribution $\mu$ over unitary operators, we say an algorithm \emph{$\left(\veps,p\right)$-average estimates} the phase with input distribution $\mu$ if
\[
\Pr\left[|\Tilde{\theta} - \theta| \leq \veps 
\colon \Tilde{\theta} \gets \cA^{c\minus U}(\ket{u}) , U\gets\mu \right] \geq p,
\]
where the probability is over the choice of $U$ and the randomness of the algorithm.
\qedsymbol
\end{definition}

Below we define the \emph{phase distinguishing problem} which can be viewed as the \emph{average-input} and \emph{discretized} version of the phase estimation problem.

\begin{definition}[Phase Distinguishing Problem]\label{def:PD}
For any $N \in\N$, let $\mathsf{PD_N}$ denote the $N$-phase distinguishing problem. Define a set of finite-dimensional unitary operators $\set{U_y}_{y\in[N]}$ in which all elements acts on $d$ qubits\footnote{That is, each $U_y$ is a unitary operator in $\C^{2^d\times2^d}$. When considering applying $U_y$ to a larger Hilbert space with more than $d$ qubits, the definition of $U_y$ naturally extends by acting as identity on the extra qubits.} and share a common eigenstate $\ket{u}$. 
For every $y\in[N]$, the unitary operator $U_y$ is defined as 
\[
U_y \coloneqq \omega_N^y\ket{u}\bra{u} + \sum_{i=2}^{2^d} \ket{e_i}\bra{e_i},
\]
where $\set{\ket{u},\ket{e_2},\dots,\ket{e_{2^d}}}$ is an orthonormal basis for $\C^{2^d}$.
\begin{itemize}
    \item Input: $N\in\N$, the eigenstate $\ket{u}$\footnote{In fact, no matter how many copies of $\ket{u}$ are given to the algorithm, the proof still holds.} and quantum oracle access to $c\minus U_y$, where $y\ugets[N]$ in the beginning.
    \item Output: An integer $y'\in[N]$.
\end{itemize}
The algorithm \emph{solves} $\mathsf{PD_N}$ if and only if $y=y'$.
We say an algorithm \emph{$p$-solves}
$\mathsf{PD_N}$ if and only if
\[
\Pr\left[y = y' \colon y'\gets \cA^{c\minus U_y}(N,\ket{u}) , y\ugets [N]\right] \geq p,
\]
where the probability is over the choice of $y$ and the randomness of the algorithm.
\qedsymbol
\end{definition}

To solve the $N$-phase distinguishing problem, it is sufficient to obtain an $\veps$-estimation of the phase with $\veps < 1/2N$ and then round the approximated phase to the closest $y/N$. As a result, we can immediately  obtain the following reductions between the problems.
\begin{lemma}\label{lem:red_1}
For every $\veps>0$ and $p\geq0$, if an algorithm $\cA$ $(\veps,p)$-estimates the phase, then the algorithm $\cA$ $(\veps,p)$-average estimates the phase with the input distribution $\mu$ defined to be uniform over the (finite) set $\set{U_y}_{y\in[N]}$ defined as in Definition~\ref{def:PD}, where $N \coloneqq \floor{1/(2\veps)}$.
\end{lemma}

\begin{lemma}\label{lem:red_2}
For every $\veps>0$ and $p\geq0$, if an algorithm $\cA$ $(\veps,p)$-average estimates the phase with the input distribution $\mu$ defined to be uniform over the (finite) set $\set{U_y}_{y\in[N]}$, where $N \coloneqq \floor{1/(2\veps)}$. Then the algorithm $\cA$ $p$-solves $\mathsf{PD_N}$.
\end{lemma}
Therefore, in the rest of the paper we focus on proving the query lower bound for solving $\mathsf{PD_N}$.
Inspired by the work of Zhandry \cite{zhandry2019record}, we provide an equivalent description of the phase distinguishing problem. First, besides the algorithm's registers $\mathsf{A}$, we augment another register $\mathsf{C}$. Notice that the sampling of $y$ can be modeled as measuring the uniform superposition state $\ket{\hat{0}}_\mathsf{C} = \frac{1}{\sqrt{N}}\sum_{y\in[N]}\ket{y}_\mathsf{C}$ in the computational bases. Furthermore, the measurement commutes with the algorithm and thus can be deferred to the end. Consequently, we define the \emph{purified view} below and analyze the experiment in the \emph{Fourier} basis.

\begin{definition}[The Purified View of the Phase Distinguishing Problem]
For any integer $N\geq 1$, $q\geq0$ and every $q$-query oracle-aided quantum algorithm $\cA$, specified by $A_0,\dots,A_q$, the \emph{purified view} of $\cA$ \wrt $\mathsf{PD_N}$ is defined as 
\[
A_q c\minus\U \dots c\minus\U A_0\ket{\phi_0}_\mathsf{A} \ket{\hat{0}}_\mathsf{C},
\]
where $c\minus\U$ is the coherent-version of $c\minus U_y$ defined as $c\minus\U\coloneqq \sum_{y\in[N]}c\minus U_y\otimes\ket{y}\bra{y}_\mathsf{C}$.
\qedsymbol
\end{definition}



\begin{lemma}\label{lem:CS_ineq} For any finite set of complex numbers $\set{\alpha_i}_{i\in[N]}$ and finite set of finite-dimensional vectors over complex numbers $\set{\ket{\psi_i}}_{i\in[N]}$, it holds that
\[
\left\|\sum_{i\in[N]} \alpha_i\ket{\psi_i}\right\|^2
\leq \sum_{i\in[N]} |\alpha_i|^2 \cdot \sum_{i\in[N]} \left\|\ket{\psi_i}\right\|^2.
\]
\end{lemma}
\begin{proof}
By the triangle inequality and the Cauchy-Schwarz inequality, we have
\[
\left\|\sum_{i\in[N]} \alpha_i\ket{\psi_i}\right\|
\le \sum_{i\in[N]} |\alpha_i| \cdot \|\ket{\psi_i}\| \le \sqrt{\sum_{i\in[N]} |\alpha_i|^2 \cdot \sum_{i\in[N]} \left\|\ket{\psi_i}\right\|^2}.
\]
Squaring both sides completes the proof.
\end{proof}

\section{Proof of the Main Theorem}\label{sec:proof}
Nayak and Wu~\cite{nayak1999quantum} proved a (tight) query lower bound for the counting problem by using the polynomial method~\cite{beals2001quantum}. The counting problem can be reduced to the phase estimation and henceforth its lower bound is obtained. 
\begin{theorem}[\cite{nayak1999quantum}]\label{thm:PE_lower}
In the phase estimation problem, for any $\veps > 0$, any oracle-aided quantum algorithm requires $\Omega(1/\veps)$ queries to output $\Tilde{\theta}$ such that $|\Tilde{\theta}-\theta|\leq\veps$ with any constant probability $>1/2$.
\end{theorem}

Note that Theorem~\ref{thm:PE_lower} only shows the \emph{worst-case} hardness in the sense that the algorithm must approximate \emph{every} inputs $U$ and $\ket{u}$ within an error $\veps$. 
In certain scenarios, the inputs might be chosen from a distribution that is publicly known prior. Moreover, the threshold of the success probability is limited to a constant greater than half. In particular, in most cryptographic contexts, the security definition requires the success probability to be negligible. As we mentioned, a natural question is that given any error tolerance $\veps>0$ and any lower bound for the success probability $p \geq 0$, what is the minimum number of queries required to achieve such an approximation? Or equivalently, given any $\veps>0$ and an upper bound on the number of queries $q \geq 0$, what is the maximum probability of outputting an $\veps$-approximation? To the best of our knowledge, it is not clear whether techniques in the previous proofs can be generalized to such settings.

In the rest of the section, we aim to prove the following main theorem.
\begin{theorem}\label{thm:main}
For any $N\in\N$ and integer $q \geq 0$, every $q$-query oracle-aided quantum algorithm solves $\mathsf{PD_N}$ with probability at most $(q+1)/N$.
In other words, for any $N\in\N$ and $p \geq 0$, every oracle-aided quantum algorithm requires at least $\Omega(pN)$ queries in order to $p$-solve $\mathsf{PD_N}$.
\end{theorem}
Then from Lemma~\ref{lem:red_1} and Lemma~\ref{lem:red_2}, we obtain the following corollary regarding the query lower bound for the standard and average-input phase estimation problem.

\begin{corollary}\label{cor:PE}
For any $\veps>0$ and $p \geq 0$, every oracle-aided quantum algorithm requires at least $\Omega(p/\veps)$ queries to $(\veps,p)$-estimate the phase.
\end{corollary}

\begin{corollary}\label{cor:AvPE}
For any $\veps>0$ and $p \geq 0$, there exists a distribution $\mu$ over unitary operators such that every oracle-aided quantum algorithm requires at least $\Omega(p/\veps)$ queries to $(\veps,p)$-average estimate the phase with input distribution $\mu$.
\end{corollary}

The rest of the section is dedicated to proving Theorem~\ref{thm:main}. The following lemma lies in the heart of the proof. Informally, in the \emph{Fourier} basis, the register $\mathsf{C}$ can be seen as a \emph{counter} that records the number of queries made by the algorithm.
In the beginning, the counter starts with the state $\ket{\hat{0}}_\mathsf{C}$ indicating zero. As an important observation, each query can increase the counter at most by $1$ (in superposition). Therefore, after $q$ queries, the possible values of the counter are between $0$ and $q$; then the final state will possess a \emph{sparse} representation in the purified view.

\begin{lemma}\label{lem:sparse}
For any $N\in\N$ and integer $q$ such that $N-1 \geq q \geq 0$, the final state of every $q$-query algorithm in the purified view \wrt $\mathsf{PD_N}$, denoted by $\ket{\psi}_\mathsf{AC}$, can be represented of the form
\[
\ket{\psi}_\mathsf{AC} = \sum_{k=0}^q \alpha_k\ket{\psi_k}_\mathsf{A}\ket{\hat{k}}_\mathsf{C},
\]
where all $\ket{\psi_k}$'s are unit (but not necessarily mutually orthogonal) vectors and $\alpha_i$'s are complex numbers satisfying the normalization condition $\sum_{k=0}^q|\alpha_k|^2=1$.
\end{lemma}
\begin{proof}
We finish the proof by induction on the number of queries $q$.
Initially, the register $\mathsf{C}$ is the zero $\ket{\hat{0}}_\mathsf{C}$ in the Fourier basis.
Note that the unitary operator performed by $\cA$ acts as identity on the register $\mathsf{C}$. Since all operators are unitary which preserves the normalization condition, the statement holds for the base case $q = 0$.

For the induction step, suppose the statement holds for some $q$. We first represent the state $\ket{\psi}_\mathsf{AC}$ in the following basis. For the qubits on which $U_y$ acts non-trivially, we choose the eigenbasis of $U_y$ that is $\set{\ket{u},\ket{e_2},\dots,\ket{e_{2^d}}}$. For the register $\sfC$, we choose the Fourier basis. Then we analyze the behavior of each basis vector. Finally, the conclusion holds by the linearity of unitary operators.

By induction hypothesis, the content of the register $\mathsf{C}$ is in $\set{\hat{0},\dots,\hat{q}}$.
An important observation is that for all $k\in [N]$, it holds that
\begin{align*}
c\minus\U\ket{1}\ket{u}\ket{\hat{k}}
& = \frac{1}{\sqrt{N}}\sum_{y\in[N]}\omega_N^{ky} c\minus\U \ket{1}\ket{u}\ket{y} 
= \ket{1}\otimes \frac{1}{\sqrt{N}}\sum_{y\in[N]}\omega_N^{ky} U_y \ket{u}\ket{y} \\
& = \ket{1}\ket{u} \sum_{y\in[N]}\frac{\omega_N^{(k+1)y}}{\sqrt{N}} \ket{y} 
= \ket{1}\ket{u}\ket{\widehat{k \scalebox{0.9}{+} 1}} 
\end{align*}
with addition modulo $N$ and the first register is the control register of $c\minus\U$. When the control register is $\ket{0}$, $c\minus\U$ simply becomes an identity.
Moreover, for any $\ket{v}$ orthogonal to $\ket{u}$, $c\minus\U$ also acts as identity on $\ket{1}\ket{v}\ket{\hat{k}}$ by the definition of $U_y$. Putting things together, we conclude that each query increases the counter $\mathsf{C}$ at most by $1$ (in superposition) in the Fourier basis. Therefore, this completes the proof.
\end{proof}

\begin{remark}
We note the proof of Lemma~\ref{lem:sparse} can be trivially extended to the setting in which the algorithm has the access to the inverse oracle $c\minus U^{-1}_y$ or power oracles $c\minus U^n_y$. The only relevant quantity is the number of possible values $k$'s of the counter that can be ``composed'' by using those oracles at most $q$ times. For example, using $c\minus U^{-1}_y$ once will subtract $1$ from the value of the counter; using $c\minus U^{n}_y$ once will add $n$ to the value of the counter.
\end{remark}

\begin{proof}[Proof of Theorem~\ref{thm:main}]
Without loss of generality, we assume that $\cA$ generates the output by measuring its output register $\mathsf{O}$ which is a part of the register $\mathsf{A}$ in the computational basis.

For every $y\in[N]$, define the projectors $\Pi_\mathsf{O}^y \coloneqq \ket{y}\bra{y}_\mathsf{O}$ and $\Pi_\mathsf{C}^y \coloneqq \ket{y}\bra{y}_\mathsf{C}$. 
Let $\ket{\psi}_\mathsf{AC}$ be the state after the whole computation and right before $\cA$ performs the final measurement to generate the output. By $\bfy$ we mean the random variable whose outcome is the choice of $y$.

The success probability of $\cA$ is given by

\begin{align*}
& \sum_{y\in[N]} \Pr\left[\cA^{c\minus U_\bfy} \text{ outputs } y \land \bfy = y \right] \\
& = \sum_{y\in[N]} \left\|\Pi_\mathsf{O}^y\otimes \Pi_\mathsf{C}^y\ket{\psi}_\mathsf{AC}\right\|^2 \\
& = \sum_{y\in[N]} \left\|\sum_{k=0}^q \alpha_k\Pi_\mathsf{O}^y\ket{\psi_k}_\mathsf{A}\otimes \Pi_\mathsf{C}^y\ket{\hat{k}}_\mathsf{C}\right\|^2 \tag{1}\label{1} \\
& = \sum_{y\in[N]} \left\|\sum_{k=0}^q \alpha_k\frac{\omega_N^{ky}}{\sqrt{N}}\Pi_\mathsf{O}^y\ket{\psi_k}_\mathsf{A}\right\|^2 \\
& \leq \sum_{y\in[N]} \left[ \left(\sum_{k=0}^q \left| \alpha_k\frac{\omega_N^{ky}}{\sqrt{N}} \right|^2\right) \cdot \left( \sum_{k=0}^q\|\Pi_\mathsf{O}^y\ket{\psi_k}_\mathsf{A}\|^2\right) \right] \tag{2}\label{2} \\
& = \frac{1}{N}\sum_{y\in[N]} \sum_{k=0}^q\|\Pi_\mathsf{O}^y\ket{\psi_k}_\mathsf{A}\|^2 \tag{3}\label{3} \\
& = \frac{1}{N}\sum_{k=0}^q\|\ket{\psi_k}_\mathsf{A}\|^2 \tag{4}\label{4} \\
& = \frac{q+1}{N} \tag{5}\label{5},
\end{align*}
where \eqref{1} is obtained by applying Lemma~\ref{lem:sparse} to $\ket{\psi}_\mathsf{AC}$; \eqref{2} follows from Lemma~\ref{lem:CS_ineq}; \eqref{3} and \eqref{5} follow from the normalization condition of Lemma~\ref{lem:sparse}; \eqref{4} follows from the fact that the projectors $\set{\Pi_\mathsf{O}^y}_{y\in[N]}$ form a complete basis.
\end{proof}

\section{An Interpretation of Phase Estimation Algorithm}\label{sec:interpretation}
With the following well-known property of EPR states in mind, one can naturally come up with the optimal algorithm from the proof in Section~\ref{sec:proof}.

\begin{fact}\label{fact:EPR}
For every $N\in \N$, the state $\ket{\Psi} \in \C^N\otimes\C^N$ defined as $\ket{\Psi} 
\coloneqq \frac{1}{\sqrt{N}}\sum_{y\in [N]}\ket{y}\ket{y}$ satisfies 
$\ket{\Psi} = \frac{1}{\sqrt{N}}\sum_{y\in [N]}\ket{\minus\hat{y}}\ket{\hat{y}}$, 
where $\ket{\minus\hat{y}}$ is the complex conjugate of $\ket{\hat{y}}$.
\end{fact}
\begin{proof}
We finish the proof by comparing the coefficients. For any $a,b\in[N]$, we have
\[
\left(\bra{a}\bra{b}\right) \left(\frac{1}{\sqrt{N}}\sum_{y\in [N]}\ket{\minus\hat{y}}\ket{\hat{y}}\right)
= \frac{1}{N\sqrt{N}}\sum_{y\in [N]}\omega_N^{-ay}\omega_N^{by}
= \frac{1}{\sqrt{N}}\delta_{a,b}
= \left(\bra{a}\bra{b}\right) \left(\frac{1}{\sqrt{N}}\sum_{y\in [N]}\ket{y}\ket{y}\right),
\]
where $\delta_{a,b}$ equals $1$ if $a = b$ and $0$ otherwise.
\end{proof}

In the above analysis, we see that in order to achieve high success probability in the phase distinguishing problem, the algorithm's output register $\mathsf{O}$ must be highly correlated with the oracle register $\mathsf{C}$. 
For intuition, consider the following state which has a success probability $1$.
\[
\frac{1}{\sqrt{N}}\sum_{y\in[N]}\ket{y}_\mathsf{O}\ket{y}_\mathsf{C}
= \frac{1}{\sqrt{N}}\sum_{y\in[N]}\ket{\minus\hat{y}}_\mathsf{O}\ket{\hat{y}}_\mathsf{C}.
\]
Notice that the values are perfectly correlated and the marginal distribution of the oracle is uniformly random.

Therefore, we can interpret the goal of the algorithm as preparing the state on the right-hand side. Combining the observation we made in Theorem~\ref{thm:main} and Fact~\ref{fact:EPR}, we show the following algorithm for solving $\mathsf{PD_N}$.
\begin{enumerate}
    \item Initial state: $\ket{0}_\sfO\ket{u}_\sfW \ket{\hat{0}}_\sfC$, where $\sfA = (\sfO,\sfW)$.
    \item Create uniform superposition: $\frac{1}{\sqrt{N}}\sum_{y\in[N]}\ket{y}_\sfO \ket{u}_\sfW \ket{\hat{0}}_\sfC$.
    \item Perform the \emph{controlled-add} defined by
    \[
    \sum_{y\in[N]} \ket{y}\bra{y}_\sfO \otimes \left(c\minus\U\right)^y
    \]
    in superposition which requires $N$ queries.
    The resulting state will be 
    \[
    \frac{1}{\sqrt{N}}\sum_{y\in[N]}\ket{y}_\sfO \ket{u}_\sfW \ket{\hat{y}}_\sfC.
    \]
    \item Perform the inverse quantum Fourier transform on the register $\sfA$:
    \[
    \frac{1}{\sqrt{N}}\sum_{y\in[N]}\ket{\minus\hat{y}}_\sfO \ket{u}_\sfW \ket{\hat{y}}_\sfC.
    \]
\end{enumerate}
The above procedures are exactly the same as the phase estimation algorithm of Cleve \etal~\cite{cleve1998quantum} but rephrased in a different perspective.
\section*{Acknowledgment}
We would like to thank Kai-Min Chung for the helpful discussions. We also thank the anonymous TQC~2023 reviewers for suggesting a simpler proof to Lemma~2.9 and useful comments.

\bibliographystyle{alpha}
\bibliography{References.bib}

\newcommand{\etalchar}[1]{$^{#1}$}
\begin{thebibliography}{CEMM98}

\bibitem[Amb02]{ambainis2002quantum}
Andris Ambainis.
\newblock Quantum lower bounds by quantum arguments.
\newblock {\em Journal of Computer and System Sciences}, 64(4):750--767, 2002.

\bibitem[AS04]{aaronson2004quantum}
Scott Aaronson and Yaoyun Shi.
\newblock Quantum lower bounds for the collision and the element distinctness
  problems.
\newblock {\em Journal of the ACM (JACM)}, 51(4):595--605, 2004.

\bibitem[BBBV97]{bennett1997strengths}
Charles~H Bennett, Ethan Bernstein, Gilles Brassard, and Umesh Vazirani.
\newblock Strengths and weaknesses of quantum computing.
\newblock {\em SIAM journal on Computing}, 26(5):1510--1523, 1997.

\bibitem[BBC{\etalchar{+}}01]{beals2001quantum}
Robert Beals, Harry Buhrman, Richard Cleve, Michele Mosca, and Ronald De~Wolf.
\newblock Quantum lower bounds by polynomials.
\newblock {\em Journal of the ACM (JACM)}, 48(4):778--797, 2001.

\bibitem[BC09]{berry2009black}
Dominic~W Berry and Andrew~M Childs.
\newblock Black-box hamiltonian simulation and unitary implementation.
\newblock {\em arXiv preprint arXiv:0910.4157}, 2009.

\bibitem[Bes05]{bessen2005lower}
Arvid~J Bessen.
\newblock Lower bound for quantum phase estimation.
\newblock {\em Physical Review A}, 71(4):042313, 2005.

\bibitem[BHMT02]{brassard2002quantum}
Gilles Brassard, Peter Hoyer, Michele Mosca, and Alain Tapp.
\newblock Quantum amplitude amplification and estimation.
\newblock {\em Contemporary Mathematics}, 305:53--74, 2002.

\bibitem[CEMM98]{cleve1998quantum}
Richard Cleve, Artur Ekert, Chiara Macchiavello, and Michele Mosca.
\newblock Quantum algorithms revisited.
\newblock {\em Proceedings of the Royal Society of London. Series A:
  Mathematical, Physical and Engineering Sciences}, 454(1969):339--354, 1998.

\bibitem[Chi10]{childs2010relationship}
Andrew~M Childs.
\newblock On the relationship between continuous-and discrete-time quantum
  walk.
\newblock {\em Communications in Mathematical Physics}, 294(2):581--603, 2010.

\bibitem[HHL09]{harrow2009quantum}
Aram~W Harrow, Avinatan Hassidim, and Seth Lloyd.
\newblock Quantum algorithm for linear systems of equations.
\newblock {\em Physical review letters}, 103(15):150502, 2009.

\bibitem[HLS07]{hoyer2007negative}
Peter Hoyer, Troy Lee, and Robert Spalek.
\newblock Negative weights make adversaries stronger.
\newblock In {\em Proceedings of the thirty-ninth annual ACM symposium on
  Theory of computing}, pages 526--535, 2007.

\bibitem[Kit95]{kitaev1995quantum}
A~Yu Kitaev.
\newblock Quantum measurements and the abelian stabilizer problem.
\newblock {\em arXiv preprint quant-ph/9511026}, 1995.

\bibitem[MNRS11]{magniez2011search}
Fr{\'e}d{\'e}ric Magniez, Ashwin Nayak, J{\'e}r{\'e}mie Roland, and Miklos
  Santha.
\newblock Search via quantum walk.
\newblock {\em SIAM journal on computing}, 40(1):142--164, 2011.

\bibitem[NC10]{nielsen2010quantum}
Michael~A Nielsen and Isaac~L Chuang.
\newblock Quantum computation and quantum information.
\newblock {\em Cambridge University Press}, 2010.

\bibitem[NW99]{nayak1999quantum}
Ashwin Nayak and Felix Wu.
\newblock The quantum query complexity of approximating the median and related
  statistics.
\newblock In {\em Proceedings of the thirty-first annual ACM symposium on
  Theory of computing}, pages 384--393, 1999.

\bibitem[Sho99]{shor1999polynomial}
Peter~W Shor.
\newblock Polynomial-time algorithms for prime factorization and discrete
  logarithms on a quantum computer.
\newblock {\em SIAM review}, 41(2):303--332, 1999.

\bibitem[Sze04]{szegedy2004quantum}
Mario Szegedy.
\newblock Quantum speed-up of markov chain based algorithms.
\newblock In {\em 45th Annual IEEE symposium on foundations of computer
  science}, pages 32--41. IEEE, 2004.

\bibitem[Zha19]{zhandry2019record}
Mark Zhandry.
\newblock How to record quantum queries, and applications to quantum
  indifferentiability.
\newblock In {\em Annual International Cryptology Conference}, pages 239--268.
  Springer, 2019.

\end{thebibliography}

\end{document}